\documentclass[11pt]{article}

\usepackage{amsmath,amsthm,amsfonts,graphicx}
\usepackage{amssymb}
\usepackage{epsfig,graphicx,color}
\usepackage{mathrsfs}
\usepackage[active]{srcltx}
\usepackage{url}
\usepackage{enumerate}
\usepackage{authblk}
\usepackage{hyperref}
\hypersetup{
colorlinks=true,
linkcolor=blue,
filecolor=blue,
citecolor=blue,  
urlcolor=blue,
}
\usepackage{comment}
\usepackage{multirow}

\usepackage{subcaption}
\usepackage[backend=biber,style=phys, maxnames=10]{biblatex}
\bibliography{rb}

\usepackage[firstpage=true]{background}

\backgroundsetup{placement=top,contents=MIT-CTP/5363,color=black,opacity=1,scale=1,position={15cm,0cm}}

\renewcommand{\H}{\mathcal{H}}

\def\>{\rangle}
\def\<{\langle}
\def\kk{\>\!\>}
\def\bb{\<\!\<}

\newcommand{\Tr}{\operatorname{Tr}}

\newtheorem{theo}{Theorem}

\newtheorem{lemma}{Lemma}

\newcommand{\E}{\mathbb{E}}

\newcommand{\id}{\operatorname{id}}

\newcommand{\cE}{\mathcal{E}}
\newcommand{\cN}{\mathcal{N}}

\newcommand{\cC}{\mathcal{C}}

\newcommand{\norm}[1]{\|#1\|}
\newcommand{\dnorm}[1]{\|#1\|_\diamond}
\DeclareMathOperator{\irr}{Irr}
\newcommand{\gend}{g_{\operatorname{end}}}
\newcommand{\G}{\mathbb{G}}
\newcommand{\cem}{\cE_{\operatorname{M}}}
\newcommand{\esp}{\cE_{\operatorname{SP}}}

\newcommand{\red}[1]{{\color{black} #1}}
\newcommand{\blue}[1]{{\color{black} #1}}

\newcommand{\B}{\mathcal{B}}
\newcommand{\M}{\mathcal{M}}
\newcommand{\m}{{\operatorname{m}}}
\newcommand{\sep}{\operatorname{sep}}
\newcommand{\inv}{\operatorname{inv}}
\newcommand{\iswap}{\mathrm{\lowercase{i}SWAP} }
\newcommand{\cnot}{\mathrm{CNOT} }
\newcommand{\ad}{\mathrm{AD} }
\newcommand{\dep}{\mathrm{dep} }
\newcommand{\xy}{\mathrm{XY} }
\newcommand{\cU}{\mathcal{U}}
\title{A framework for randomized benchmarking over compact groups}
\author[1]{Linghang Kong\thanks{linghang@mit.edu}}
\affil[1]{Center for Theoretical Physics, MIT, Cambridge, MA 02139, United States}
\date{}

\begin{document}
\maketitle
\begin{abstract}
    Characterization of experimental systems is an essential step in developing and improving quantum hardware. A collection of protocols known as Randomized Benchmarking (RB) was developed in the past decade, which provides an efficient way to measure error rates in quantum systems. In a recent paper~\cite{helsen2020}, a general framework for RB was proposed, which encompassed most of the known RB protocols and overcame the limitation on error models in previous works. However, even this general framework has a restriction: it can only be applied to a finite group of gates. This does not meet the need posed by experiments, in particular the demand for benchmarking non-Clifford gates and continuous gate sets on quantum devices. In this work we generalize the RB framework to continuous groups of gates and show that as long as the noise level is reasonably small, the output can be approximated as a linear combination of matrix exponential decays. As an application, we numerically study the fully randomized benchmarking protocol (i.e. RB with the entire unitary group as the gate set) enabled by our proof. This provides a unified way to estimate the gate fidelity for any quantum gate in an experiment.
\end{abstract}

\section{Introduction}
The past decade has witnessed great developments of large-scale quantum computers. Besides improving techniques for building quantum hardware, significant efforts have been devoted to developing tools for characterizing them. An essential set of such tools is collectively known as \emph{Randomized Benchmarking (RB)}, whose central idea is to generate sequences of random gates with varying lengths and study the dependence of the overall noise level on the sequence length. The main advantages of RB over other benchmarking tools include: 
\begin{enumerate}
    \item Robustness against state preparation and measurement (SPAM) error
    \item Efficiency, in the sense that most RB schemes require polynomial amount of resources in the number of qubits.
\end{enumerate}
Here we will first introduce some of the best-known RB protocols. Standard RB is a protocol originally proposed in~\cite{emerson2005,levi2007}, which generates sequences of Haar random unitaries followed by an inverting gate that would evolve the system back to the initial state. It was shown that the survival probability (i.e. the probability of getting the initial state after the gate sequence) decays exponentially with the sequence length, and the decay rate represents the error rate per gate. Later it was showed in~\cite{dankert2009,magesan2011} that as long as the gates are sampled from a unitary 2-design (i.e. the first two moments of the gate distribution matches that of the Haar measure), an exponential decay can be observed. A common choice of a 2-design is the uniform distribution over the Clifford group, and the corresponding RB protocol is also known as Clifford RB. A related protocol is interleaved randomized benchmarking~\cite{magesan2012}, which generates sequences that are interleaved with some specific gate in addition to the original random sequences. From the change in error rate caused by the addition of the interleaved gate, one can determine the gate fidelity. Note that when the gate set forms a group, it is often required that the interleaved gate lie in the gate set, so that the ending gate is also contained in the gate set (there are exceptions, e.g.~\cite{erhard2019,harper2017}, but they are designed for specific cases and cannot be generalized). For example, interleaved Clifford RB can only be used for estimating the fidelity of Clifford gates. 

All these protocols rely on the exponential decaying behavior of the output, which depends on the noise model. In early works of RB (see~\cite{emerson2005} for example), the error model is assumed to be gate-independent, which is often too stringent for practical experimental systems. In recent years, gate dependent noise has been analyzed based on Fourier transformation over groups~\cite{hashagen2018,wallman2018,merkel2021}. A recent paper~\cite{helsen2020} proposed a general framework of RB, which states that the output result can be approximated by a linear combination of matrix exponential decays even when the noise channel depends on the gate, as long as the error rate is small enough. Also, this framework encompasses most of the known RB protocols, which provides a unified way of analyzing them. 

However, one limitation of this result is that the gate set in this framework has to be a finite group, which makes it difficult to benchmark non-Clifford gates and continuous gate sets in a unified way. But such gate sets play essential roles in practical quantum computing. Non-Clifford gates are necessary for universal quantum computing, as Clifford gates alone are efficiently simulable classically and cannot generate arbitrary quantum gates. Continuous gate sets are often implemented as the elementary gates in experiments (for example $\xy$-gates~\cite{abrams2020} and fSim-gates~\cite{foxen2020}), and they are frequently used in near-term quantum algorithms such variational quantum eigensolvers~\cite{peruzzo2014} and QAOA~\cite{farhi2014}. This poses the need for a protocol to benchmark such gates in a unified way.

In order to accommodate the need for benchmarking a general quantum gate, in this work we will extend the uniform-probability framework in~\cite[Thm.~8]{helsen2020} to a gate set that forms a compact Lie group. See our Thm.~\ref{thm:main} for the exact statement. As an application, this allows us to do \emph{fully randomized benchmarking (FRB)}, namely RB on the set of all unitary operators, whose interleaved version could measure the gate fidelity of any quantum gate. The main difference between a finite group and a compact group is that a compact group has infinitely many irreducible representations. To prove our result we generalize a matrix perturbation theorem (Thm.~6 in~\cite{helsen2020}) to an infinite dimensional Hilbert space setting and then apply it to the infinite dimensional Fourier space.

Besides proving the main theorem, we also conduct a numerical study of FRB, whose validitiy is established by our main result. We compared the output stability of FRB and Clifford RB on benchmarking tasks that can be done by both protocols, and found that FRB has comparable stability as Clifford RB. We also studied the performance of FRB on a benchmarking task beyond the capability of Clifford RB, namely estimating the fidelity of gates in a continuous family using interleaved FRB. We found that the result given by the numerical simulation of FRB matches the theoretical prediction pretty well.

The paper is structured as follows. After introducing known results about RB, in Section~\ref{sec:main} we will present our main result, which is a generalization of~\cite[Thm.~8]{helsen2020}. Section~\ref{sec:numerics} is a numerical study of FRB. Finally in Section~\ref{sec:discussion} we conclude the paper with a discussion.

\subsection{Preliminaries}
\label{sec:pre}

For two quantum states $\rho$ and $\sigma$, their fidelity is defined as
\begin{equation}
    F(\rho, \sigma) = \|\sqrt\rho\sqrt\sigma\|_1^2 = \left(\Tr\sqrt{\sqrt{\rho}\sigma\sqrt{\rho}}\right)^2.
\end{equation}
We can then define the average gate fidelity based on this state fidelity. Suppose that a quantum gate $U$ is implemented as a quantum channel $\cU(\cdot)$, then the average gate fidelity $F_g$ is defined as the average fidelity between the ideal and actual output state when the input state is Haar random,
\begin{equation}
    F_g = \E_{|\psi\>\sim \text{Haar}}F(U|\psi\>\<\psi|U^\dagger, \cU(|\psi\>\<\psi|)).
\end{equation}
Suppose that $\cU$ is given by the application of $U$ followed by a noise channel $\Lambda$, i.e.
\begin{equation}
    \cU(\rho) = \Lambda(U\rho U^\dagger),
\end{equation}
and that $\Lambda$ has Kraus operators $A_i$,
\begin{equation}
    \Lambda(\rho) = \sum_k A_k \rho A_k^\dagger,
\end{equation}
then one can know that~\cite{emerson2005}
\begin{equation}
    F_g = \frac{\sum_k |\Tr A_k|^2 + D}{D^2 + D} \label{eq:fidelity}
\end{equation}
where $D$ is the dimension of the system. $F_g$ depends only on the error channel $\Lambda$ and is independent of the gate $U$.

For the standard RB protocol, it is known that the survival probability $p_m$ after a length-$m$ gate sequence takes the form~\cite{emerson2005}
\begin{equation}
    p_m = \alpha^m A + B,\quad \alpha = \frac{\sum_k |\Tr A_k|^2 - 1}{D^2 - 1}
\end{equation}
if the noise is gate independent. Then by finding $\alpha$ through numerical fitting, one could calculate the average gate fidelity of the gate set by
\begin{equation}
    F_g = \alpha + \frac{1-\alpha}{D}.
\end{equation}

Interleaved RB~\cite{magesan2012} works similarly. Let $\alpha$ and $\alpha'$ be the decay rate of survival probability for the original and interleaved sequences respectively. Then the gate fidelity for the interleaved gate is given by
\begin{equation}\label{eq:irb}
    F_g = 1 - \frac{D-1}{D}\left(1 - \frac{\alpha'}{\alpha}\right).
\end{equation}

\section{Main Results}
\label{sec:main}

\subsection{Description of the Framework}
\label{sec:framework}
Here we will briefly describe our RB framework. The only difference between ours and the one in~\cite[Thm.~8]{helsen2020} is that the gate set is allowed to be a general compact Lie group. An RB protocol based on this framework is defined by the following parameters:
\begin{itemize}
    \item A gate set $\G$, from which the gate sequences are generated. $\G$ can be a finite group or a general compact Lie group. The gates will be sampled uniformly from the Haar measure over $\G$, which is the probability distribution $\nu$ such that $\E_{g\sim \nu} f(g) = \E_{g\sim\nu} f(gh)$ for any $h\in\G$ and any function $f$ over $\G$, and is equal to the uniform distribution when $\G$ is finite.
    \item A reference implementation $\omega(\cdot)$, which is a mapping from $\G$ to superoperators on the quantum system. On the linear space of all operators on the quantum system, $\omega$ is a representation of $\G$. Here we require that the quantum system is finite-dimensional for simplicity.
    \item An ending gate $\gend \in \G$, which is the overall effect of the gate sequence. In other words, every generated gate sequences $g_1,\ldots,g_t$ should satisfy $g_t\cdot \ldots \cdot g_2 \cdot g_1 = \gend$, where $\cdot$ is the multiplication in the group. In most RB protocols $\gend$ is the identity.
    \item A set of sequence lengths $\mathbb{M}$. Note that in our convention, the sequence length $m\in\mathbb{M}$ refers to the number of random gates and does not include the final inversion gate.
    \item An input state of the quantum system.
    \item An output POVM $\{\Pi_i\}_i$ that is applied to the system after each gate sequence.
\end{itemize}

In an experiment, the implementation of group element will be denoted by $\phi(g)$, which is a noisy version of $\omega(g)$. The RB experiment is run with the following steps:
\begin{enumerate}
    \item For each $m\in\mathbb{M}$ do the following:
    \begin{enumerate}
        \item Initialize the quantum system.
        \item Sample $g_1, g_2,\ldots, g_m$ uniformly from $\G$ and apply $\phi(g_1),\phi(g_2),\ldots, \phi(g_m)$ to the system.
        \item Calculate $g_{\inv} = (g_m\cdot \ldots \cdot g_2 \cdot g_1)^{-1}$ and apply $\phi(\gend \cdot g_{\inv})$ to the system.
        \item Apply POVM $\{\Pi_i\}$ to get the measurement result.
        \item Repeat the experiment above many times with fixed $g_1,\ldots,g_m$ to estimate the output probability $p(i|g_1,\ldots, g_m)$.
        \item Repeat for many different samples of $g_1,\ldots,g_m$ and take average of $p(i|g_1,\ldots, g_m)$ to calculate $p(i,m)$.
    \end{enumerate}
    \item Do a numerical exponential decay fitting for $p(i,m)$
\end{enumerate}

By Thm.~\ref{thm:main}, $p(i,m)$ can be approximated by a linear combination of matrix exponential decays, and by numerical fitting one can know the error rate of the system.

\subsection{Proof}
In our theorem, the distance between superoperators is quantified using diamond norms $\|\cdot\|_\diamond$, defined by the following equation
\begin{equation}
    \|\cC_L\|_\diamond = \sup_{R;|\psi\>_{LR}} \|(\cC_L \otimes \id_R)(|\psi\>\<\psi|_{LR})\|_1
\end{equation}
where the supremum is taken over all reference systems $R$ and all normalized states $|\psi\>$. 

Our main result is Thm.~8 of~\cite{helsen2020} generalized to a compact group and can be stated as follows.

\begin{theo}\label{thm:main}
	Consider an RB experiment defined with compact Lie group $\G$ and reference representation $\omega = \bigoplus_{\lambda\in \Lambda} \sigma_\lambda^{\oplus n_\lambda}$. We denote by $\phi$ the implementation map, i.e. the operation acting on the quantum system. For a specific sequence length $m>0$, let $p(i,m)$ be the outcome probability for measurement result $i$. 
	If there exists $\delta>0$ such that
	\begin{equation}
	\int dg \dnorm{\omega(g) - \phi(g)} \leq \delta < \frac{1}{9},
	\end{equation}
	then the RB output probability $p(i,m)$ satisfies
	\begin{equation}
	\left|p(i,m) -\sum_{\lambda \in \Lambda} \Tr(A_\lambda (M_\lambda)^m)\right| \leq \frac{16}{1-9\delta}\left[\delta\left(2+ \frac{4\delta}{1-5\delta}\right)\right]^m.
	\end{equation}
	where $M_\lambda$ and $A_\lambda$ are $n_\lambda\times n_\lambda$ matrices, and $M_\lambda$ only depends on the implementation $\phi$.

\end{theo}

We need the following lemma about perturbation theory on an infinite dimensional Hilbert space, which is a generalization of Thm.~6 in~\cite{helsen2020}. Note we will be working in a Hilbert space that can potentially be of infinite dimension, but we will use block matrix notation for calculations involving multiple operators. This is fine by defining 
\begin{equation}
    [A, B]
    \begin{pmatrix}
        C & D \\
        E & F
    \end{pmatrix}
    = [G, H]
    \longleftrightarrow
    G=AC+BE, H = AD + BF
\end{equation}
and
\begin{equation}
    [L_1, L_2]^\dagger A [R_1, R_2] =
    \begin{pmatrix}
    A_{11} & A_{12}  \\
    A_{21} & A_{22}
    \end{pmatrix}
    \longleftrightarrow
    A_{ij} = L_i^\dagger A R_j, \forall i,j\in\{1,2\}. 
\end{equation}
We will also use the $\sep(\cdot)$ function. Consider a Hilbert space $\H$ and let $V_1$ and $V_2$ be orthogonal subspaces of $\H$. Let $A_1$ and $A_2$ be bounded operators on $V_1$ and $V_2$ respectively. Let $\M(V_1,V_2)$ be the set of all operators $P$ that map $V_2$ to $V_1$, i.e. $P=X_1 P X_2$ where $X_1$ and $X_2$ are projectors into $V_1$ and $V_2$. For a norm $\|\cdot\|$ defined on $\H$, we define the corresponding $\sep(\cdot)$ function as
\begin{equation}
    \sep(A_1,A_2) = \inf_{P\in \M(V_1, V_2)} \frac{\|A_1 P - P A_2\|}{\|P\|}.
\end{equation}
The function $\sep(\cdot)$ is known to be stable against perturbations
\begin{equation}
    |\sep(A_1 + E_1, A_2 + E_2) - \sep(A_1, A_2)| \le \|E_1\| + \|E_2\|.
\end{equation}
\begin{lemma}\label{lem:perturb}
    Consider a Hilbert space $\H$ with norm $\|\cdot\|$. This norm naturally induces a norm on operators.
    Let $V_1$ be a \blue{finite dimensional} subspace of $\H$ and let $V_2$ be its orthogonal complement. Let $X_1$ and $X_2$ be the projectors into $V_1$ and $V_2$ respectively. Then we define $A_1 = A= X_1$ and $A_2 = 0$, so $\sep(A_1, A_2) = \sep(A_2, A_1) = 1$ and $A$ is block diagonal w.r.t. $X_1$ and $X_2$,
    \[
        [X_1, X_2]^\dagger A [X_1, X_2] =
        \begin{pmatrix}
            A_1 & 0 \\
            0 & A_2
        \end{pmatrix}.
    \]
    Let $E$ be an arbitrary operator. If $E$ has the following properties
    \begin{align}
     \mathrm{sep}(A_1,A_2)- \norm{X_1 E X_1} - \norm{X_2 E X_2} >& 0 \label{eq:sep} \\
    \frac{\norm{X_1 EX_2} \norm{X_2 E X_1}}{\big( \mathrm{sep}(A_1,A_2)- \norm{X_1 E X_1} - \norm{X_2 E X_2}\big)^2 } <& \frac{1}{4} \label{eq:pert_props}
    \end{align}
        then there exist operators $P_1$ and $P_2$ that satisfies
    \begin{align}
     X_2 P_1 X_1 &= P_1, X_1 P_2 X_2 = P_2, \\
    \norm{P_1} &\leq \frac{\red{2}\norm{X_2 E X_1}}{\mathrm{sep}(A_1,A_2) - \norm{X_1 E X_1} - \norm{X_2 E X_2}},\\
    \norm{P_2} &\leq \frac{\norm{X_2 E X_1}}{\mathrm{sep}(A_1,A_2) - \norm{X_1 E X_1 + X_1EX_2P_1} - \norm{X_2 E X_2- P_1X_1EX_2}}
    \end{align}
    and can turn $A+E$ into a block diagonal structure in the following way:
    \begin{equation}
    [L_1, L_2]^\dagger (A+E)[R_1,R_2] = \begin{pmatrix} A'_1 & 0 \\ 0 & A'_2\end{pmatrix}
    \end{equation}
    where
    \begin{align}
    [R_1,R_2] = [X_1,X_2] \begin{pmatrix} X_1 & 0 \\ P_1 & X_2\end{pmatrix}\begin{pmatrix} X_1 & P_2 \\ 0 & X_2\end{pmatrix} 
    \\
    [L_1,L_2]^\dagger = \begin{pmatrix} X_1 & -P_2 \\ 0 & X_2\end{pmatrix}\begin{pmatrix} X_1 & 0 \\ -P_1 & X_2\end{pmatrix}[X_1,X_2]^\dagger,
    \end{align}
    and
    \begin{align}
        \red{A'_1 = A_1 + X_1 E X_1  + X_1 E X_2 P_1}, \\
        A'_2 = A_2 + X_2 E X_2  -P_1 X_1 E X_2.
    \end{align}
Equivalently,
    we have
    \begin{equation}
    A+E = R_1 A'_1 L_1^\dagger + R_2 A'_2 L_2^\dagger.
    \end{equation}
\end{lemma}
The proof can be found in Appendix~\ref{app:perturb}.

Consider a compact group $\G$, and let $\irr(\G)$ be the infinite set of all (finite-dimensional) irreducible representations of $\G$. For $\lambda \in \irr(\G)$, $d_\lambda$ is defined to be the dimension of representation $\lambda$. For a matrix-valued function $\phi$ defined on $\G$, the Fourier operator is defined as
\begin{equation}
    F(\phi) \equiv \int dg \bigoplus_{\lambda\in \irr(\G)} \bar\sigma_\lambda(g) \otimes \phi(g) \equiv \int dg \bar\omega_\G(g) \otimes \phi(g).
\end{equation}
where $\bar\omega_\G(g) = \bigoplus_{\lambda\in \irr(\G)} \bar\sigma_\lambda(g)$.
The function $\phi$ could be retrieved from $F(\phi)$ via the inverse Fourier transformation
\begin{equation}
    \phi(g)= \Tr_{V_{\omega_\G}}[F(\phi)D_\G (\bar\omega_\G(g^{-1}) \otimes I)]
\end{equation}
where
\begin{equation}
    D_\G = \bigoplus_{\lambda\in \irr(\G)} d_\lambda I.
\end{equation}

In the case we study, $\phi(g)$ is a superoperator on some Hilbert space, and its diamond norm induces two norms on Fourier operators
\begin{align*}
    \|F(\phi)\|_{\max} =& \max_g \|\Tr_{V_{\omega_\G}}[F(\phi)D_\G (\bar\omega_\G(g^{-1}) \otimes I)]\|_\diamond = \max_g\|\phi(g)\|_\diamond, \\
    \|F(\phi)\|_{\m} =& \E_{g\sim G} \|\Tr_{V_{\omega_\G}}[F(\phi)D_\G (\bar\omega_\G(g^{-1}) \otimes I)]\|\|_\diamond = \E_{g \sim G}\|\phi(g)\|_\diamond. \\
\end{align*}
They are known to satisfy the following inequalities\cite{helsen2020}
\begin{align}
    \|F(\phi) F(\phi')\|_{\max} \le& \|F(\phi) \|_{\max}\|F(\phi')\|_{\max}, \nonumber \\ 
    \|F(\phi) F(\phi')\|_{\m} \le& \|F(\phi) \|_{\m}\|F(\phi')\|_{\m}, \nonumber \\
    \|F(\phi) F(\phi')\|_{\max} \le& \|F(\phi) \|_{\max}\|F(\phi')\|_{\m}. \label{eq:norms}
\end{align}

One fact we will use in the following derivation is the $\m$-norm of identity operator.
\begin{lemma}
\label{lem:id}
$\|\id\|_\m=1$ where $\id$ is the identity operator in the Fourier space.
\end{lemma}
\begin{proof}
Using inverse Fourier transformation we have
\begin{align}
    \|\id\|_m=&\E_g\left\|\Tr_{V_{\omega_\G}}[\id \times D_\G (\bar\omega_\G(g^{-1}) \otimes I)]\right\|_\diamond \nonumber\\
    =& \E_g \sum_\sigma d_\sigma \Tr[\sigma(g^{-1})] \|I\|_\diamond \nonumber\\
    =& \E_g \delta(g-I) \|I\|_\diamond \nonumber\\
    =& 1,
\end{align}
where we used the fact~\cite[Above Eq.~(4.15)]{witten1992} that
\begin{equation}
    \sum_\sigma d_\sigma \Tr[\sigma(g^{-1})] = \delta(g-I).
\end{equation}
\end{proof}

\begin{theo}[restatement of Thm.~\ref{thm:main}]
	Consider an RB experiment defined with compact Lie group $\G$ and reference representation $\omega = \bigoplus_{\lambda\in \Lambda} \sigma_\lambda^{\oplus n_\lambda}$. We denote by $\phi$ the implementation map, i.e. the operation acting on the quantum system. For a specific sequence length $m>0$, let $p(i,m)$ be the outcome probability for measurement result $i$. 
	If there exists $\delta>0$ such that
	\begin{equation}\label{eq:norm_assump}
	\int dg \dnorm{\omega(g) - \phi(g)} \leq \delta < \frac{1}{9},
	\end{equation}
	then the RB output probability $p(i,m)$ satisfies
	\begin{equation}
	\left|p(i,m) -\sum_{\lambda \in \Lambda} \Tr(A_\lambda (M_\lambda)^m)\right| \leq \frac{16}{1-9\delta}\left[\delta\left(2+ \frac{4\delta}{1-5\delta}\right)\right]^m.
	\end{equation}
	where $M_\lambda$ and $A_\lambda$ are $n_\lambda\times n_\lambda$ matrices, and $M_\lambda$ only depends on the implementation $\phi$.

\end{theo}
\begin{proof}
We will work in the linear space of operators, so $|A\kk$ refers to the vectorized version of operator $A$, and $\bb A | B \kk =  \Tr[A^\dagger B]$. We take into consideration the state preparation and measurement (SPAM) error, so the actual initial state is denoted by $\esp(\rho_0)$ and the actual measurement is denoted by $\cem(\Pi_i)$, where $\rho_0$ is the initial state and $\{\Pi_i\}_i$ is the final POVM. First note that $p(i,m)$ could be expressed as a convolution,
\begin{equation}
    p(i,m) = \bb \cem(\Pi_i)|\phi^{*(m+1)}(\gend)|\esp(\rho_0)\kk
\end{equation}
where $\gend$ is the ending gate. Since the convolution of operators can be expressed as the multiplication of their Fourier transformation, we have
\begin{equation}
    p(i,m) = \bb \cem(\Pi_i)|\Tr_{V_{\omega_\G}}[F(\phi)^{m+1}D_\G (\bar\omega_\G(\gend^{-1}) \otimes I)]|\esp(\rho_0)\kk. \label{eq:pim}
\end{equation}

The next step is to use Lemma~\ref{lem:perturb} on $F(\phi) = F(\omega)+F(\phi-\omega)$. $F(\omega)$ is known to be a projector of finite dimension, and $F(\phi-\omega)$ is viewed as a perturbation. We set $X_1 = F(\omega)$ and $X_2 = \id - F(\omega)$. The condition Eq.~\eqref{eq:norm_assump} implies
\begin{equation}
    \|F(\phi-\omega)\|_m \le \delta,
\end{equation}
so Eq.~\eqref{eq:sep} and Eq.~\eqref{eq:pert_props} now becomes
\begin{align}
    1-\|X_1 F(\phi-\omega)\|_\m - \|X_2 F(\phi-\omega)X_2\|_\m \ge& 1 - 5\|F(\phi-\omega)\|_\m >0, \\
    \frac{\|X_1 F(\phi-\omega) X_2\|_\m \|X_2 F(\phi-\omega) X_1\|_\m}{(1-\|X_1 F(\phi-\omega)\|_\m - \|X_2 F(\phi-\omega)X_2\|_\m)^2} \le& \frac{(2\delta)^2}{(1-5\delta)^2} < \frac{1}{4},
\end{align}
where we used the fact that $\|X_1\|_\m \le 1$, $\|X_2\|_\m\leq \|\id\|_\m+\|X_1\|_\m=2$. Then by Lemma~\ref{lem:perturb}
\begin{align}
    F(\phi) =& R_1(F(\omega) + X_1 F(\phi-\omega) X_1  + X_1 F(\phi-\omega) X_2 P_1)L_1^\dagger \nonumber\\
    &+ R_2(X_2 F(\phi-\omega) X_2  -P_1 X_1 F(\phi-\omega) X_2)L_2^\dagger \nonumber \\
    =& R_1(F(\omega) + X_1 F(\phi-\omega)(X_1+ X_2 P_1))L_1^\dagger + R_2 (X_2  -P_1 X_1) F(\phi-\omega) X_2 L_2^\dagger.
\end{align}
Here, we have used the fact that $X_1$ and $X_2$ are Hermitian, and that $X_1 F(\omega) X_1 = F(\omega), X_2 F(\omega) X_2 = 0$. Since $L^\dagger = R^{-1}$, we have $L_2^\dagger R_1 = 0, L_1^\dagger R_2 = 0$, so
\begin{equation}
    F(\phi)^{m+1} = R_1[F(\omega) + X_1 F(\phi-\omega)(X_1+ X_2 P_1)]^{m+1}L_1^\dagger + R_2 [(X_2  -P_1 X_1) F(\phi-\omega) X_2]^{m+1} L_2^\dagger
\end{equation}
and correspondingly Eq.~\eqref{eq:pim} could be split into 2 terms, $p(i,m)=\alpha+\beta$ where
\begin{align}
    \alpha =& \bb \cem(\Pi_i)|\Tr_{V_{\omega_\G}}[R_1[F(\omega) + X_1 F(\phi-\omega)(X_1+ X_2 P_1)]^{m+1}L_1^\dagger D_\G (\bar\omega_\G(\gend^{-1}) \otimes I)]|\esp(\rho_0)\kk, \\
    \beta =& \bb \cem(\Pi_i)|\Tr_{V_{\omega_\G}}[R_2 [(X_2  -P_1 X_1) F(\phi-\omega) X_2]^{m+1} L_2^\dagger D_\G (\bar\omega_\G(\gend^{-1}) \otimes I)]|\esp(\rho_0)\kk.
\end{align}

To calculate $\alpha$, note that $F(\omega)$ and $F(\phi-\omega)$ are both block diagonal, so we could choose the operators $X_1$, $X_2$, $P$, $L$ and $R$ to be block diagonal as well. We use the superscript ${}^{(\lambda)}$ to represent the block corresponding to block $\lambda$. Note that for any $\lambda$ not contained in $\omega$, $F(\omega)^{(\lambda)} = X_1^{(\lambda)}=0$, so we only need to consider $\lambda\in\Lambda$. Then
\begin{align}
    \alpha =& \bb \cem(\Pi_i)|\Tr_{V_{\omega_\G}}[R_1[F(\omega) + X_1 F(\phi-\omega)(X_1+ X_2 P_1)]^{m+1}L_1^\dagger D_\G (\bar\omega_\G(\gend^{-1}) \otimes I)]|\esp(\rho_0)\kk, \nonumber\\
    =& \bb \cem(\Pi_i)|\Tr_{V_{\omega_\G}}[D_\G(\bar\omega_\G(\gend^{-1}) \otimes I)F(\phi)R_1[F(\omega) + X_1 F(\phi-\omega)(X_1+ X_2 P_1)]^{m}L_1^\dagger  ]|\esp(\rho_0)\kk, \nonumber \\
    =& \bb \cem(\Pi_i)|\Tr_{V_{\omega_\G}}[D_\G(\bar\omega_\G(\gend^{-1}) \otimes I)F(\phi)R_1[X_1 F(\phi)(X_1+ X_2 P_1)]^{m}L_1^\dagger  ]|\esp(\rho_0)\kk, \nonumber\\
    =& \sum_{\lambda\in\Lambda} d_\lambda\bb \cem(\Pi_i)|\Tr_{V_{\lambda}}\left[(\bar\omega_\lambda(\gend^{-1}) \otimes I)F(\phi)^{(\lambda)}R_1^{(\lambda)}\left[X_1^{(\lambda)} F(\phi)^{(\lambda)}(X_1^{(\lambda)}+ X_2^{(\lambda)} P_1^{(\lambda)})\right]^{m}L_1^{(\lambda)\dagger}  \right]|\esp(\rho_0)\kk \nonumber\\
    =& \sum_{\lambda\in\Lambda} d_\lambda \Tr\left[L_1^{(\lambda)\dagger}(I_{V_\lambda} \otimes |\esp(\rho_0)\kk\bb \cem(\Pi_i)|^{(\lambda)})(\bar\omega_\lambda(\gend^{-1}) \otimes I)F(\phi)^{(\lambda)}R_1^{(\lambda)}\left[X_1^{(\lambda)} F(\phi)^{(\lambda)}(X_1^{(\lambda)}+ X_2^{(\lambda)} P_1^{(\lambda)})\right]^{m}  \right], 
\end{align}
Here $|\esp(\rho_0)\kk\bb \cem(\Pi_i)|^{(\lambda)}$ is obtained by projecting $|\esp(\rho_0)\kk\bb \cem(\Pi_i)|$ into block $\lambda$. The operator $|\esp(\rho_0)\kk\bb \cem(\Pi_i)|$ does not have to be block diagonal, but the off-diagonal part does not contribute to the result when contracting with a block diagonal operator.

Note that $P_1 X_1 = P_1$, we have
\begin{equation}
    X_1^{(\lambda)}\left[X_1^{(\lambda)} F(\phi)^{(\lambda)}(X_1^{(\lambda)}+ X_2^{(\lambda)} P_1^{(\lambda)})\right]X_1^{(\lambda)} = \left[X_1^{(\lambda)} F(\phi)^{(\lambda)}(X_1^{(\lambda)}+ X_2^{(\lambda)} P_1^{(\lambda)})\right],
\end{equation}
which means that $X_1^{(\lambda)} F(\phi)^{(\lambda)}(X_1^{(\lambda)}+ X_2^{(\lambda)} P_1^{(\lambda)})$ lies within the $n_\lambda$-dimensional subspace defined by $X_1^{(\lambda)}$. Then we can define
\begin{align}
    M_\lambda\equiv& X_1^{(\lambda)} F(\phi)^{(\lambda)}(X_1^{(\lambda)}+ X_2^{(\lambda)} P_1^{(\lambda)}), \\
    A_\lambda\equiv& d_\lambda X_1^{(\lambda)}\left[L_1^{(\lambda)\dagger}(|\esp(\rho_0)\kk\bb \cem(\Pi_i)|\otimes I_{V_\lambda})(\bar\omega_\lambda(\gend^{-1}) \otimes I)F(\phi)^{(\lambda)}R_1^{(\lambda)}\right]X_1^{(\lambda)}
\end{align}
and have
\begin{equation}
    \alpha = \sum_{\lambda\in\Lambda}\Tr[A_\lambda M_\lambda^m].
\end{equation}
With a suitable choice of basis $A_\lambda$ and $M_\lambda$ could be written as $n_\lambda$-dimensional matrices.

To give a bound for $\beta$ we use Eq.~\eqref{eq:norms} and have
\begin{align}
    \beta \le& \left\|\Tr_{V_{\omega_\G}}[R_2 [(X_2  -P_1 X_1) F(\phi-\omega) X_2]^{m+1} L_2^\dagger D_\G (\bar\omega_\G(\gend^{-1}) \otimes I)]\right\|_\diamond \nonumber \\
    \le& \left\|R_2 [(X_2  -P_1 X_1) F(\phi-\omega) X_2]^{m+1} L_2^\dagger \right\|_{\max} \nonumber \\
    =& \left\|(X_2 + X_1 P_2 + X_2 P_1 P_2) [(X_2  -P_1 X_1) F(\phi-\omega) X_2]^{m+1} (X_2 - P_1 X_1) \right\|_{\max} \nonumber \\
    \le& \|(X_2 + X_1 P_2 + X_2 P_1 P_2)(X_2  -P_1 X_1)\|_\m \| F(\phi-\omega) \|_{\max} \nonumber \\
    &\times \|[X_2(X_2  -P_1 X_1) F(\phi-\omega)]^m\|_\m \|X_2(X_2 - P_1 X_1)\|_\m \nonumber \\
    \le& (2+\|P_2\|_\m + \|P_1P_2\|_\m)(2+\|P_1\|_\m)^2\|F(\phi-\omega)\|_{\max}[(2+\|P_1\|_\m)\|F(\phi-\omega)\|_\m]^m
\end{align}
where we used the fact that $\|X_1\|_\m \le 1$, $\|X_2\|_\m \le \|\id\|_\m + \|X_1\|_\m = 2$ ($\|\id\|_\m \le 1$ using Lemma~\ref{lem:id}).
By Lemma~\ref{lem:perturb} we have
\begin{align}
    \|P_1\|_\m \le& \frac{2\|X_1 F(\phi-\omega) X_2\|_\m}{1-\|X_1 F(\phi-\omega)X_1\|_\m - \|X_2 F(\phi-\omega)X_2\|_\m} \nonumber \\
    \le& \frac{4\delta}{1-5\delta}, \\
    \|P_2\|_\m \le& \frac{\|X_1 F(\phi-\omega) X_2\|_\m}{1-\|X_1 F(\phi-\omega)X_1\|_\m - \|X_2 F(\phi-\omega)X_2\|_\m - 2\|P_1\|_\m \|X_1 F(\phi-\omega) X_2\|_\m} \nonumber \\
    \le& \frac{2\delta}{1-5\delta-4\delta\frac{4\delta}{2-5\delta}} \nonumber\\
    =&\frac{2\delta(1-5\delta)}{(1-\delta)(1-9\delta)}.
\end{align}
Note that $\|F(\phi-\omega)\|_{\max} \le \|F(\phi)\|_{\max} + \|F(\omega)\|_{\max} \le 2$, we have
\begin{equation}
    \beta \le \frac{16(1-3\delta)^2(1-8\delta)}{(1-5\delta)^2(1-9\delta)}\left(\delta\left[2+ \frac{4\delta}{1-5\delta}\right]\right)^m\le \frac{16}{1-9\delta}\left(\delta\left[2+ \frac{4\delta}{1-5\delta}\right]\right)^m.
\end{equation}
\end{proof}

\section{Numerical Tests Of The Protocol}
\label{sec:numerics}
In this section we will demonstrate some numerical results about our fully randomized benchmarking (FRB) protocol and its interleaved version. For an experimental system with Hilbert space dimension $d$ ($d=4$ in our case), FRB is RB over the group $SU(d)$, i.e. the group of all unitary operators on the Hilbert space. The reference implementation is the conjugation of the unitary operator, i.e. for $U\in SU(d)$, the superoperator $\omega(U)$ maps any input state $\rho$ to $U\rho U^\dagger$. The interleaved version of FRB could estimate the average fidelity of any gate $V$ by running an ``interleaved sequence'' on the device in additional to the original version of FRB. The interleaved sequence applies the gate $V$ after each randomly sampled gate. After doing the numerical fitting, the gate fidelity is given by Eq.~\eqref{eq:irb}. As was noted in~\cite{helsen2020}, interleaved RB could fit into the framework in Section~\ref{sec:framework} by taking the following implementation map:
\begin{equation}
    \phi_C(g) = \phi(V) \circ \phi(V^{-1}g)
\end{equation}
where $\phi$ is the original implementation map on the device and $\circ$ is the composition of superoperators.

In Section~\ref{sec:compare} we will run some tasks that are also doable by the traditional Clifford RB, and we will compare the numerical stability of the two protocols quantified by the standard deviation of the result over multiple runs. In Section~\ref{sec:iswap} we will run interleaved FRB on a continuous family of gates and compare the obtained gate fidelity with the corresponding theoretical results.

All RB simulations are implemented using Qiskit~\cite{qiskit}. For each RB experiment we pick 11 different sequence lengths, and 30 random circuits will be generated for each given length. Each circuit will be run 5000 times in order to determine the output distribution. The average error rate per gate will be calculated based on the output distribution of the $11\times 30$ circuits described above.

For all the numerical results, the gates in the generated gate sequence will be decomposed into some elementary gates before the simulation. The noise will then be applied to each elementary gate during the simulation.

\subsection{Comparison With Clifford Randomized Benchmarking}
\label{sec:compare}
In this section the basis gate in our simulator will be $\cnot$ and single qubit gates. The single qubits gates are further classified into $U_1, U_2$ and $U_3$ defined as
\begin{equation}
    U_1(\lambda) =
    \begin{pmatrix}
    1 & 0 \\
    0 & e^{i\lambda}
    \end{pmatrix},
    U_2(\phi,\lambda) =
    \frac{1}{\sqrt 2}
    \begin{pmatrix}
    1 & -e^{i\lambda} \\
    e^{i\phi} & e^{i(\phi +\lambda)}
    \end{pmatrix},
    U_3(\theta,\phi,\lambda) =
    \begin{pmatrix}
    \cos(\theta/2) & -e^{i\lambda}\sin(\theta/2) \\
    e^{i\phi}\sin(\theta/2) & e^{i(\phi+\lambda)}\cos(\theta/2)
    \end{pmatrix}.
\end{equation}
$U_1$ is a $Z$ rotation, while $U_2$ (resp. $U_3$) can be implemented by $Z$ rotations together with 1 (resp. 2) rotations in $X$ or $Y$ directions. $Z$ rotations are often implemented by a shift of basis and has no error. So in our error model we assume that $U_1$ is noiseless, while the noise for $U_3$ is twice as strong as that of $U_2$.

The noise models we consider are the depolarization channel $\cN_\dep^p$ and the amplitude damping channel $\cN_\ad^\gamma$. For a general $n$-qubit system, the depolarization channel is defined as
\begin{equation}
    \cN_\dep^p(\rho) = (1-p)\rho + p \frac{I}{2^n} = (1-p)\rho + \frac{p}{4^n} \sum_{\mathbf{i}\in\{0,1,2,3\}^n} \sigma_\mathbf{i} \rho \sigma_\mathbf{i}, \label{eq:dep}
\end{equation}
where $\sigma_\mathbf{i}$ is the Pauli operators on $n$ qubits. For a single qubit, the amplitude damping channel is defined as
\begin{equation}
    \cN_\ad^\gamma(\rho) = A_0\rho A_0^\dagger + A_1\rho A_1^\dagger, \quad 
    A_0 =
    \begin{pmatrix}
    1 & 0 \\
    0 & \sqrt{1-\gamma} \\
    \end{pmatrix}, \quad
    A_1 =
    \begin{pmatrix}
    0 & \sqrt{\gamma} \\
    0 & 0
    \end{pmatrix}. \label{eq:ad}
\end{equation}

We first use RB to study the error rate per gate on this model, with the noise model being depolarization and amplitude damping respectively. For depolarization channel, we assume $p=0.01$ for $\cnot$ gates, $p=0.002$ for $U_2$ and $p=0.004$ for $U_3$. For amplitude damping channel, we assume $\gamma=0.01$ on both qubits for $\cnot$ gates, $\gamma=0.003$ for $U_2$ and $\gamma=0.006$ for $U_3$.

We first run Clifford and FRB 30 times each and compare their standard deviation, which can represent the stability of the fitted result. The numerical results are summarized in Table~\ref{tab:rb}. The discrepancy in the mean is to be expected, as Clifford RB gives the average error rate per Clifford gate while FRB gives the error rate per Haar random gate. A random 2-qubit Clifford gate could be decomposed into 1.5 $\cnot$ gate on average\cite[Supplemental Material]{corcoles2013}, while a Haar random gate needs 3 $\cnot$ gates\cite{zhang2004}, so it is reasonable that FRB gives a larger error rate than Clifford RB. The standard deviation for FRB is slightly larger than that of Clifford RB, but in all cases it is small compared with the mean.
\begin{table}[ht]
    \centering
        \begin{tabular}{c|c|c|c|c}
            \hline
             & \multicolumn{2}{c|}{Depolarization} & \multicolumn{2}{c}{Amplitude Damping} \\
            \cline{2-5}
             & Clifford RB & FRB & Clifford RB & FRB \\
            \hline 
            Mean & 0.0159 & 0.0407 & 0.0167 & 0.0421 \\
            \hline
            Standard Deviation & $1.36\times 10^{-4}$ & $1.61 \times 10^{-4}$ & $2.90\times 10^{-4}$ & $5.29\times 10^{-4}$ \\
            \hline
        \end{tabular}
    \caption{Comparison of Clifford RB and FRB on depolarizing channel and amplitude damping channel.}
    \label{tab:rb}
\end{table}

We also run similar tests on interleaved benchmarking, with $\cnot$ gate as the interleaved gate. Again Clifford RB and FRB are run 30 times each, and the mean and stand deviation of the obtained gate fidelity are listed in Table~\ref{tab:irb}. By Eqs.~\eqref{eq:fidelity} and~\eqref{eq:dep}, 
\begin{equation}
    F_g = \frac{(1-p)4^n + p + 2^n}{2^n(2^n+1)} = 0.9925
\end{equation}
for our depolarization channel, and by Eqs.~\eqref{eq:fidelity} and~\eqref{eq:ad},
\begin{equation}
    F_g = \frac{(1+\sqrt{1-\gamma})^4 + 2^n}{(2^n+1)2^n} = 0.9920
\end{equation}
for amplitude damping channel (note that the actual error channel is amplitude damping acting on both of the qubits). We can see that the mean of the numerical results matches the theoretical values of the gate fidelity. Again the standard deviation of FRB is slightly larger than that of Clifford RB, but it is small enough for the experiment.
\begin{table}[ht]
    \centering
        \begin{tabular}{c|c|c|c|c}
            \hline
             & \multicolumn{2}{c|}{Depolarization} & \multicolumn{2}{c}{Amplitude Damping} \\
            \cline{2-5}
             & Clifford RB & FRB & Clifford RB & FRB \\
            \hline 
            Mean & 0.9925 & 0.9925 & 0.9922 & 0.9921 \\
            \hline
            Standard Deviation & $1.58\times 10^{-4}$ & $1.88 \times 10^{-4}$ & $4.37\times 10^{-4}$ & $4.84\times 10^{-4}$\\
            \hline
        \end{tabular}
    \caption{Comparison of interleaved Clifford RB and FRB on depolarizing channel and amplitude damping channel.}
    \label{tab:irb}
\end{table}

\subsection{Benchmarking Of $\iswap$ Gate Family}
\label{sec:iswap}

In this section we will use interleaved FRB to estimate the gate fidelity of a continuous family of gates known as the $\iswap$ gate family~\cite{majer2007,wendin2007,steffen2006,bialczak2010,neeley2010,dewes2012,sillanpaa2007}. Since most of the gates in this family are not Clifford, such task is not doable with the traditional Clifford RB.

The gates in the $\iswap$ gate family are also known as $\xy$-gates. We follow the convention in~\cite{abrams2020} and define it as
\begin{equation}
    \xy(\theta) = 
    \begin{pmatrix}
        1 & 0 & 0 & 0 \\
        0 & \cos(\theta/2) & i\sin(\theta/2) & 0 \\
        0 & i\sin(\theta/2) & \cos(\theta/2) & 0 \\
        0 & 0 & 0 & 1
    \end{pmatrix},
\end{equation}
and in particular, $\xy(\pi)$ is known as the $\iswap$ gate.
This gate family could be implemented as the time evolution of $XX+YY$ type interaction on the two qubits, i.e.
\begin{equation}
    \xy(\theta) = e^{-i\theta H_{XY}},\quad H_{XY}\equiv -\frac{1}{4}(XX+YY).
\end{equation}
Note that $\xy(\theta)$ are distinct gates for $0\le \theta < 4\pi$, but in most cases we are only interested in the parameter range $0 \le \theta \le \pi$, as a gate with $\pi<\theta<4\pi$ could be decomposed into single-qubit $Z$ rotations and a gate within this range.

For our simulation, each gate in the random sequence will be decomposed into elementary gates, which include single-qubit gates and $\iswap$. The error on elementary gates include thermal relaxation error and coherent error.
\begin{itemize}
    \item Thermal relaxation error is defined using parameters $T_1$, $T_2$ (relaxation time) and $T_g$ (gate time). When $T_2 < T_1$, the error channel for a single qubit takes the simple form
    \begin{equation}
        \cN(\rho) = (1-p_r - p_z) \rho + p_z Z \rho Z + p_r \Tr[\rho] |0\>\<0| \label{eq:thermal}
    \end{equation}
    where
    \[
        p_r = 1 - e^{-\frac{T_g}{ T_1}}, \quad p_z = (1 - p_r)(1-e^{-\frac{T_g}{T_2} + \frac{T_g}{T_1}}).
    \]
    For our simulation, we assume that the thermal relaxation error happens on the two qubits independently. We take $T_1 = 100\text{ms}$ and $T_2 = 20\text{ms}$. For single-qubit gates, $T_g=20\text{ns}$ and for $\iswap$, $T_g=60\text{ns}$. Since the $\xy$ gates are implemented as time evolution of $H_{XY}$, the gate time is assumed to be proportional to $\theta$. More specifically, for $\xy(\theta)$ we assume $T_g = \frac{\theta}{\pi} \times 60\text{ns}$.
    \item Coherent error is defined as the application of a non-identity unitary to the system, due to the imprecision of the control parameters. When generating $XY$ gates using Hamiltonian $H_{XY}$, then evolution time might be slightly larger than expected, and the Hamiltonian might have a small $ZZ$ coupling. We assume that the actual unitary applied is
    \begin{equation}
        \widetilde{\xy}(\theta) = \exp(-i(\theta+\delta\theta) (H_{XY} + \delta z H_{ZZ})) = \exp(-i\delta\theta H_{XY} -i(\theta+\delta\theta) \delta z H_{ZZ})\xy(\theta) \label{eq:coherent}
    \end{equation}
    where
    \[
        H_{ZZ} = -\frac{1}{4}ZZ.
    \]
    For our simulation we take $\delta\theta = 0.01$, $\delta z = 0.01$.
\end{itemize}

The gate fidelity could be calculated theoretically using Eqs.~\eqref{eq:fidelity},~\eqref{eq:thermal}, and~\eqref{eq:coherent}, and we will compare it with numerical results in Fig.~\ref{fig:compare}. We can see that the output of our protocol is consistent with the theoretical value.

\begin{figure}[ht]
     \centering
     \begin{subfigure}[b]{0.6\textwidth}
         \centering
         \includegraphics[width=\textwidth]{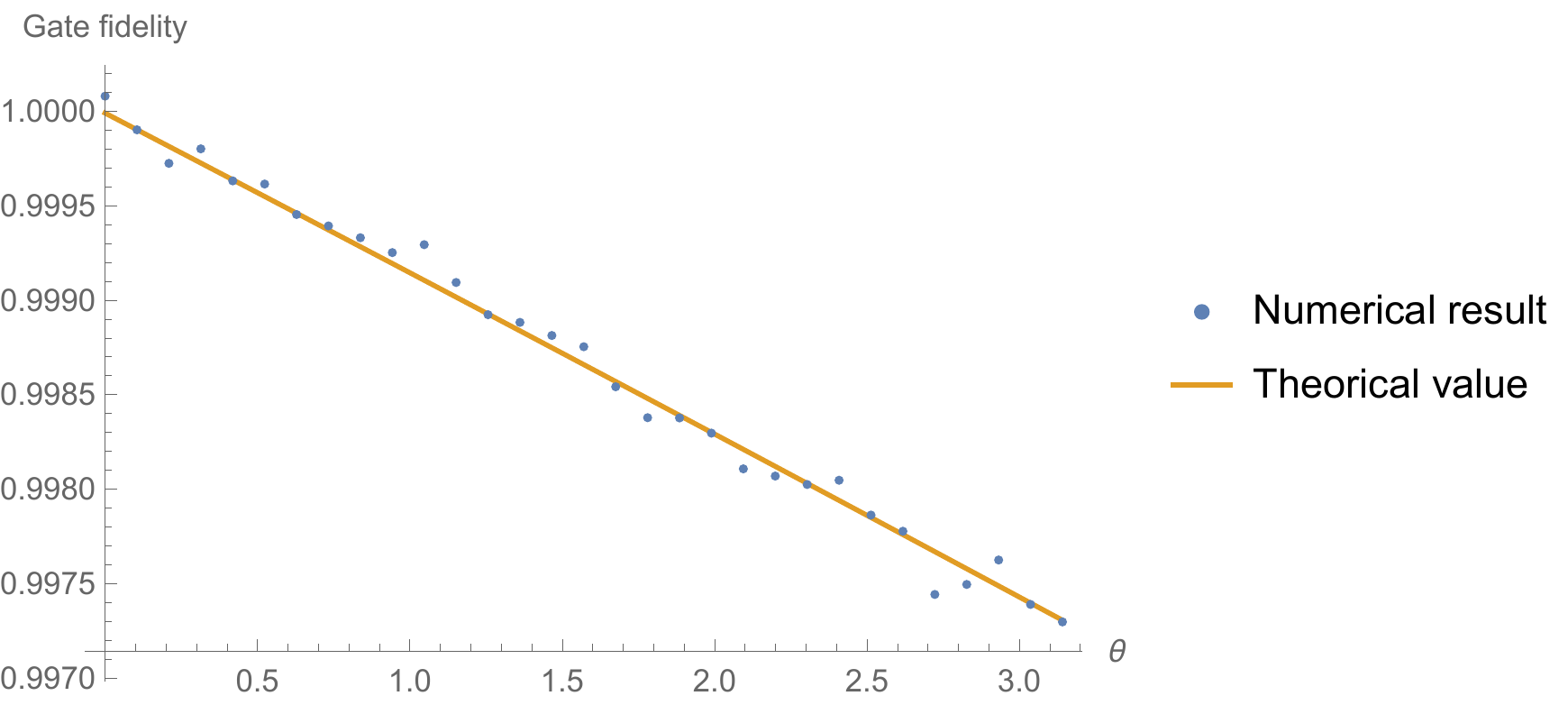}
         \caption{}
         \label{fig:compare}
     \end{subfigure}
     \hfill
     \begin{subfigure}[b]{0.6\textwidth}
         \centering
         \includegraphics[width=\textwidth]{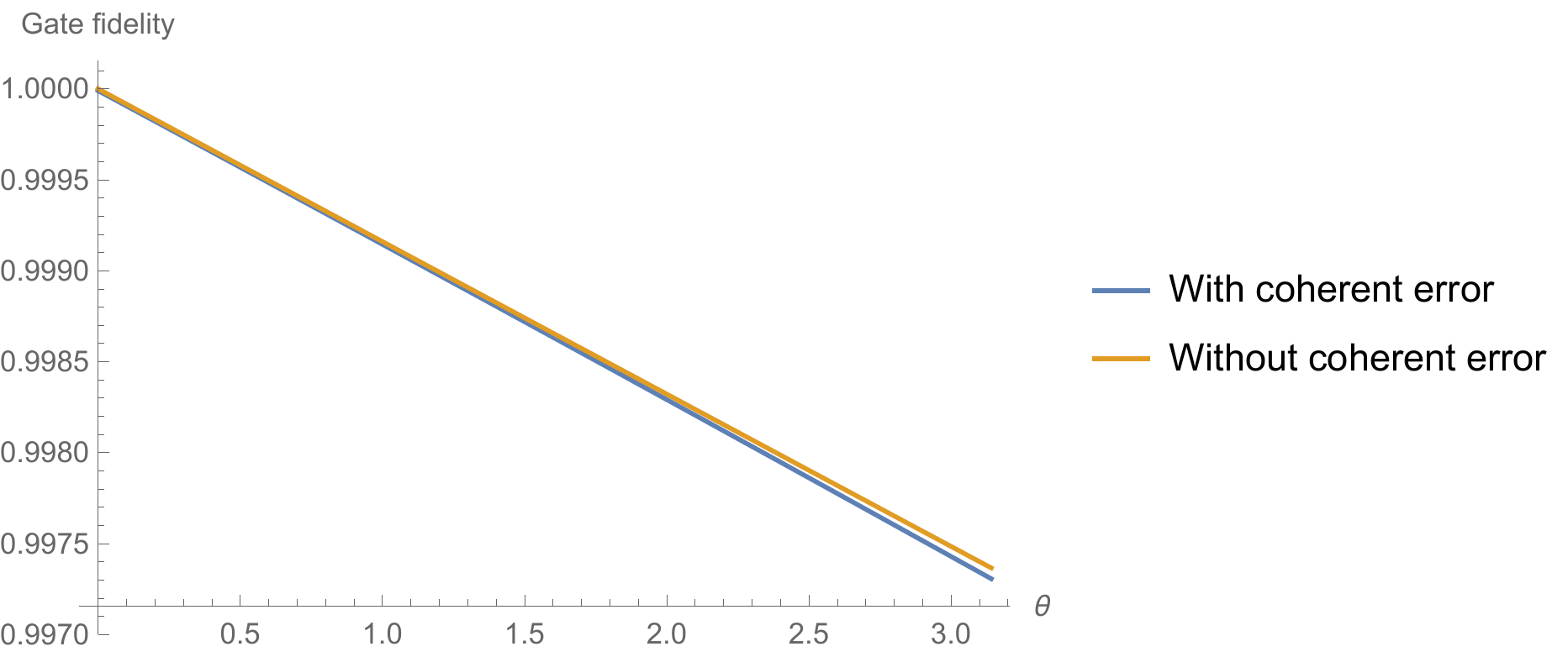}
         \caption{}
         \label{fig:coherent}
     \end{subfigure}
    \caption{(a) Comparison of the numerical results against the theoretical calculation of gate fidelity. (b) Theoretical value of gate fidelity with or without coherent error.}
    \label{fig:figure}
\end{figure}

As a side note, we find that the coherent error make little contribution towards the overall error rate, as can be seen in Fig.~\ref{fig:coherent}. This could also be confirmed by Eq.~\eqref{eq:fidelity} that the decrease of fidelity due to coherent errors is second order small. This shows that gate fidelity might not be a good indicator for coherent errors, which might be of independent interest.

\section{Discussions}
\label{sec:discussion}
In the work we have generalized the RB framework in~\cite{helsen2020} to a continuous group gate set, which encompasses most of the known RB protocols and ensures an exponential decay form of result as long as the average error rate of gate implementation is small enough. As an application one could use FRB as a unified way to estimate gate fidelity for any gate. It would be interesting to see other applications of this framework.

One should note that FRB is not exactly scalable, because the resources needed for generating a Haar random unitary on $n$ qudits is exponential in $n$. But still FRB is very useful for benchmarking gates on a small number of qubits.

Another possible direction is to further weaken the assumptions on noise. The error could be gate dependent in our model, but it is assumed to stay the same each time the gate is implemented. In practical quantum devices, however, the parameters might drift over time, so this assumption might not hold rigorously. An interesting future direction would be to take this into account and see what the RB output would be in this setting.

\section*{Acknowledgement}
The author would thank Jianxin Chen, Dawei Ding, Cupjin Huang, Liang Xiao and Qi Zhang for helpful discussions and feedback. LK is supported by NSF grant CCF-1452616.

\appendix
\section{Proof Of Perturbation Theory}
\label{app:perturb}

We need the following lemma~\cite[Thm.~2.11, Chapter V]{stewart1990}.
\begin{lemma}\label{lem:eqn}
Consider a Banach space $\B$ with norm $\|\cdot\|$. This norm naturally induces a norm for linear operators on $\B$. Let $T$ be a bounded linear operator with bounded inverse, and set
\[
    \delta = \|T^{-1}\|^{-1}.
\]
Let $\phi:\B \to \B$ be a function that satisfies
\begin{equation}
    \|\phi(x)\| \le \eta \|x\|^2
\end{equation}
and 
\begin{equation}
    \|\phi(x) - \phi(y)\| \le 2\eta \max\{\|x\|,\|y\|\}\|x-y\|,\quad \forall x, y \in \B
\end{equation}
for some constant $\eta > 0$. For any $g\in B$, let $\gamma = \|g\|$. If $4\gamma\eta/\delta^2 < 1$ then there exists a unique solution $x\in \B$ for the equation
\begin{equation}
    Tx = g+\phi(x),
\end{equation}
and $x$ satisfies
\begin{equation}
    \|x\| \le \frac{2\gamma}{\delta + \sqrt{\delta^2-4\gamma\eta}} < 2\frac{\gamma}{\delta}.
\end{equation}
\end{lemma}

In the following we will discuss a Hilbert space that can potentially be of infinite dimension, and we will use block matrix notation for calculations involving multiple operators. More specifically,
\begin{equation}
    [A, B]
    \begin{pmatrix}
        C & D \\
        E & F
    \end{pmatrix}
    = [G, H]
    \longleftrightarrow
    G=AC+BE, H = AD + BF
\end{equation}
and
\begin{equation}
    [L_1, L_2]^\dagger A [R_1, R_2] =
    \begin{pmatrix}
    A_{11} & A_{12}  \\
    A_{21} & A_{22}
    \end{pmatrix}
    \longleftrightarrow
    A_{ij} = L_i^\dagger A R_j, \forall i,j\in\{1,2\}. 
\end{equation}

Another notation we will use is the $\sep(\cdot)$ function. Consider a Hilbert space $\H$ and let $V_1$ and $V_2$ be orthogonal subspaces of $\H$. Let $A_1$ and $A_2$ be bounded operators on $V_1$ and $V_2$ respectively. Let $\M(V_1,V_2)$ be the set of all operators $P$ that map $V_2$ to $V_1$, i.e. $P=X_1 P X_2$ where $X_1$ and $X_2$ are projectors into $V_1$ and $V_2$. For a norm $\|\cdot\|$ defined on $\H$, we define the corresponding $\sep(\cdot)$ function as
\begin{equation}
    \sep(A_1,A_2) = \inf_{P\in \M(V_1, V_2)} \frac{\|A_1 P - P A_2\|}{\|P\|}.
\end{equation}
The function $\sep(\cdot)$ is known to be stable against perturbations
\begin{equation}
    |\sep(A_1 + E_1, A_2 + E_2) - \sep(A_1, A_2)| \le \|E_1\| + \|E_2\|.
\end{equation}
It could be calculated using the following theorem.
\begin{lemma}
\label{lem:invertible}
Suppose that $A_1$ is invertible on $V_1$ and satisfies $\|A_1^{-1}\|\|A_2\|< 1,\sep(A_1,A_2)>0$,  then the linear operator $T(P) \equiv A_1 P - P A_2$ is invertible and
\[
    \sep(A_1, A_2) = \|T^{-1}\|^{-1}.
\]
\end{lemma}
{Remark: It would be nice if the lemma could be proved with the condition $\sep(A_1,A_2)>0$ only, though the current version is sufficient for our purpose.}
\begin{proof}
We first show that the linear operator $T$ is injective. Suppose otherwise that there exists some nonzero $P\in \M(V_1, V_2)$ such that $T(P)=0$, then we have $\sep(A_1, A_2)=0$, which contradicts the condition $\sep(A_1, A_2)>0$.

Then we need to show that for any $Y\in \M(V_1, V_2)$, there exists $P \in \M(V_1, V_2)$ such that $T(P)=Y$. Consider the following recursive relation
\begin{equation}
    P_{k+1} = A_1^{-1}(Y+P_k A_2)
\end{equation}
with an arbitrarily chosen $P_0 \in \M(V_1, V_2)$. We have
\begin{equation}
    \|P_{k+1}-P_k\| = \|A_1^{-1}(P_k-P_{k-1})A_2\| \le \|A_1^{-1}\|\|A_2\| \|P_k-P_{k-1}\|,
\end{equation}
and since $\|A_1^{-1}\|\|A_2\|< 1$, $\{P_k\}$ is a Cauchy sequence and must have a limit $P$. This limit $P$ satisfies $P = A_1^{-1}(Y+P A_2)$, which is equivalent to $T(P)=Y$.

Thus we have shown that $T$ is a bijection, and is therefore invertible. Then
\begin{equation}
    \|T^{-1}\| = \sup_{P\in \M(V_1, V_2)} \frac{\|T^{-1}(P)\|}{\|P\|} = \sup_{P\in \M(V_1, V_2)} \frac{\|P\|}{\|T(P)\|} = \sup_{\substack{P\in \M(V_1, V_2)\\ \|P\|=1}} \frac{1}{\|T(P)\|} = \left(\inf_{\substack{P \in \M(V_1, V_2)\\ \|P\|=1}}\|T(P)\|\right)^{-1},
\end{equation}
which means that
\begin{equation}
    \|T^{-1}\|^{-1} = \inf_{\substack{P\in \M(V_1, V_2)\\ \|P\|=1}}\|T(P)\| = \sep(A_1, A_2).
\end{equation}

\end{proof}

For the cases that we will study, we have the following lemma about invertibility of linear operators.
\begin{lemma}\label{lem:invertible2}
    Suppose $V_1$ is a subspace of $\H$ with finite dimension and let $A_1$ be the projector onto it. Let $V_2$ be the complement of $V_1$ and the operator $A_2=0$. Consider linear operators $\Delta_1$ and $\Delta_2$ on $V_1$ and $V_2$. If $\Delta_1$ and $\Delta_2$ satisfy the following condition
    \[
        1 - \|\Delta_1\| - \|\Delta_2\| > 0,
    \]
    then the linear superoperator $T(P) \equiv (A_1+\Delta_1) P - P (A_2+\Delta_2)$ is invertible.
\end{lemma}
\begin{proof}
    For any nonzero state $|\psi\>\in V_1$, we have
    \begin{equation}
        \|(A_1 + \Delta_1)|\psi\>\| \ge \|A_1 |\psi\>\| - \|\Delta_1 |\psi\>\| \ge \||\psi\>\| - \|\Delta_1\|\||\psi\>\|=(1-\|\Delta_1\|)\||\psi\>\|>0,
    \end{equation}
    so $A_1 - \Delta_1$ is an injection in $V_1$. Since $V_1$ is of finite dimension, $A_1-\Delta_1$ is also invertible on $V_1$. Additionally we have
    \begin{equation}
        \|(A_1 + \Delta_1)^{-1}\| = \sup_{|\psi\>} \frac{\|(A_1 + \Delta_1)^{-1}|\psi\>\|}{\||\psi\>\|} = \sup_{|\psi\>} \frac{\||\psi\>\|}{\|(A_1 + \Delta_1)|\psi\>\|} \le \sup_{|\psi\>}\frac{\||\psi\>\|}{(1 - \|\Delta_1\|)\||\psi\>\|} = \frac{1}{1-\|\Delta_1\|},
    \end{equation}
    so
    \begin{equation}
        \|(A_1 + \Delta_1)^{-1}\| \|A_2 + \Delta_2\| \le \frac{\|\Delta_2\|}{1-\|\Delta_1\|}< 1,
    \end{equation}
    and by Lemma~\ref{lem:invertible} we know that the superoperator $T$ is invertible.
\end{proof}

\begin{lemma}[Recap of Lemma~\ref{lem:perturb}]
    Consider a Hilbert space $\H$ with norm $\|\cdot\|$. This norm naturally induces a norm on operators.
    Let $V_1$ be a \blue{finite dimensional} subspace of $\H$ and let $V_2$ be its orthogonal complement. Let $X_1$ and $X_2$ be the projectors into $V_1$ and $V_2$ respectively. Then we define $A_1 = A= X_1$ and $A_2 = 0$, so $\sep(A_1, A_2) = \sep(A_2, A_1) = 1$ and $A$ is block diagonal w.r.t. $X_1$ and $X_2$,
    \[
        [X_1, X_2]^\dagger A [X_1, X_2] =
        \begin{pmatrix}
            A_1 & 0 \\
            0 & A_2
        \end{pmatrix}.
    \]
    Let $E$ be an arbitrary operator. If $E$ has the following properties
    \begin{align}
     \mathrm{sep}(A_1,A_2)- \norm{X_1 E X_1} - \norm{X_2 E X_2} >& 0 \label{eq:sep2} \\
    \frac{\norm{X_1 EX_2} \norm{X_2 E X_1}}{\big( \mathrm{sep}(A_1,A_2)- \norm{X_1 E X_1} - \norm{X_2 E X_2}\big)^2 } <& \frac{1}{4} \label{eq:pert_props2}
    \end{align}
        then there exist operators $P_1$ and $P_2$ that satisfies
    \begin{align}
     X_2 P_1 X_1 &= P_1, X_1 P_2 X_2 = P_2, \\
    \norm{P_1} &\leq \frac{\red{2}\norm{X_2 E X_1}}{\mathrm{sep}(A_1,A_2) - \norm{X_1 E X_1} - \norm{X_2 E X_2}},\\
    \norm{P_2} &\leq \frac{\norm{X_2 E X_1}}{\mathrm{sep}(A_1,A_2) - \norm{X_1 E X_1 + X_1EX_2P_1} - \norm{X_2 E X_2- P_1X_1EX_2}}
    \end{align}
    and can turn $A+E$ into a block diagonal structure in the following way:
    \begin{equation}\label{eq:matrix_2_blocks}
    [L_1, L_2]^\dagger (A+E)[R_1,R_2] = \begin{pmatrix} A'_1 & 0 \\ 0 & A'_2\end{pmatrix}
    \end{equation}
    where
    \begin{align}
    [R_1,R_2] = [X_1,X_2] \begin{pmatrix} X_1 & 0 \\ P_1 & X_2\end{pmatrix}\begin{pmatrix} X_1 & P_2 \\ 0 & X_2\end{pmatrix} \label{eq:block}
    \\
    [L_1,L_2]^\dagger = \begin{pmatrix} X_1 & -P_2 \\ 0 & X_2\end{pmatrix}\begin{pmatrix} X_1 & 0 \\ -P_1 & X_2\end{pmatrix}[X_1,X_2]^\dagger,
    \end{align}
    and 
    \begin{align}
        \red{A'_1 = A_1 + X_1 E X_1  + X_1 E X_2 P_1}, \\
        A'_2 = A_2 + X_2 E X_2  -P_1 X_1 E X_2.
    \end{align}
    Equivalently, we have
    \begin{equation}
    A+E = R_1 A'_1 L_1^\dagger + R_2 A'_2 L_2^\dagger.
    \end{equation}
\end{lemma}
\begin{proof}
Since $A$ is block diagonal, we have
\begin{equation}
    [X_1, X_2]^\dagger (A+E) [X_1, X_2] =
    \begin{pmatrix}
        A_1+ X_1 E X_1  & X_1 E X_2 \\
        X_2 E X_1 & A_2 + X_2 E X_2
    \end{pmatrix}.
\end{equation}
The first step is to find $P_1$ such that
\begin{equation}
    \begin{pmatrix} X_1 & 0 \\ -P_1 & X_2\end{pmatrix} [X_1, X_2]^\dagger (A+E) [X_1, X_2] \begin{pmatrix} X_1 & 0 \\ P_1 & X_2\end{pmatrix} =
    \begin{pmatrix}
        A_1''  & Y \\
        0 & A_2''
    \end{pmatrix}.
\end{equation}
for operators $A_1''=A_1 + X_1 E X_1 + X_1 E X_2 P_1, A_2''=A_2 + X_2 E X_2 - P_1 X_1 E X_2$ and $Y=X_1 E X_2$. Equivalently, $P_1$ satisfies the following equation
\begin{equation}
    T(P_1) = X_2 E X_1 - P_1 X_1 E X_2 P_1 \label{eq:p1}
\end{equation}
with
\begin{equation}
    \quad T(P_1) := P_1(A_1 + X_1 E X_1) - (A_2 + X_2 E X_2)P_1.
\end{equation}
Consider the Banach space $\B$ of all possible $P_1$ such that $P_1 = X_2 P_1 X_1$, i.e. the set $\M(V_2, V_1)$. Using Eq.~\eqref{eq:sep2} and Lemma~\ref{lem:invertible2}, the linear operator $T$ is invertible.

Then we have
\begin{equation}
    \frac{\|X_1 E X_2\| \|X_2 E X_1\|}{\sep(A_1 + X_1 E X_1, A_2 + X_2 E X_2)^2} \le \frac{\|X_1 E X_2\| \|X_2 E X_1\|}{(\sep(A_1, A_2)-\|X_1 E X_1\|-\|X_2 E X_2\|)^2} < \frac{1}{4},
\end{equation}
and by Lemma~\ref{lem:eqn}, Eq.~\eqref{eq:p1} has a unique solution with
\begin{equation}
    \|P_1\| \le \frac{2\|X_2E X_1\|}{\sep(A_1 + X_1 E X_1, A_2 + X_2 E X_2)} \le \frac{2\|X_2E X_1\|}{\sep(A_1, A_2)-\|X_1 E X_1\|-\|X_2 E X_2\|}.
\end{equation}

Similarly we want to find $P_2$ such that
\begin{equation}
    \begin{pmatrix} X_1 & -P_2 \\ 0 & X_2\end{pmatrix}
    \begin{pmatrix}
        A_1''  & Y \\
        0 & A_2''
    \end{pmatrix}
    \begin{pmatrix} X_1 & P_2 \\ 0 & X_2\end{pmatrix} =
    \begin{pmatrix}
        A_1' & 0 \\
        0 & A_2'
    \end{pmatrix}.
\end{equation}
We have $A_1' = A_1'', A_2'=A_2''$ and $P_2$ satisfies
\begin{equation}
    -P_2 A_2'' + A_1'' P_2 + Y = 0,
\end{equation}
or equivalently
\begin{equation}
    X_1 E X_2 = T'(P_2),
\end{equation}
where
\begin{equation}
    T'(P_2) = P_2(A_2 + X_2 E X_2 - P_1 X_1 E X_2) - (A_1 + X_1 E X_1 + X_1 E X_2 P_1)P_2.
\end{equation}

Note that
\begin{align}
     & \sep(A_1, A_2)  - \norm{X_1 E X_1 + X_1EX_2P_1} - \norm{X_2 E X_2- P_1X_1EX_2} \nonumber\\
    \ge & \sep(A_1, A_2) - \|X_1EX_1\|-\|X_2 E X_2\|-2\|P_1\| \|X_1EX_2\|\nonumber \\
    \ge & \sep(A_1, A_2) - \|X_1EX_1\|-\|X_2 E X_2\|- \frac{2 \|X_1EX_2\| \times2\|X_2E X_1\|}{\sep(A_1, A_2)-\|X_1 E X_1\|-\|X_2 E X_2\|} \nonumber \\
    >& 0
\end{align}
where we used Eq.~\eqref{eq:pert_props2} in the final step. By Lemma~\ref{lem:invertible2} we can see that $T'$ is invertible. Then $P_2=T'{}^{-1}(X_1EX_2)$ and satisfies
\begin{equation}
    \|P_2\| \le \|T'{}^{-1}\| \|X_1EX_2\| \le \frac{\|X_1EX_2\|}{\sep(A_1,A_2) - \norm{X_1 E X_1 + X_1EX_2P_1} - \norm{X_2 E X_2- P_1X_1EX_2}}.
\end{equation}
\end{proof}

\printbibliography

\end{document}